\documentclass[conference]{IEEEtran}
\IEEEoverridecommandlockouts
\ifCLASSINFOpdf
\else
\fi
\usepackage[top=0.75in, bottom=0.75in, left=0.75in, right=0.75in]{geometry}
\usepackage{bbm}
\usepackage{verbatim}
\usepackage{graphicx}
\usepackage{cite}
\usepackage{url}
\usepackage[cmex10]{amsmath}
\usepackage{amssymb}
\usepackage{amsmath}
\usepackage{amsthm}
\usepackage{algorithm, algorithmicx, algpseudocode}
\usepackage[caption=false,font=footnotesize]{subfig}
\usepackage{color}
\usepackage{cite}
\usepackage{epstopdf}
\usepackage{calc}
\usepackage{array}
\usepackage{stfloats}
\usepackage{siunitx}
\usepackage{mathtools}
\usepackage{optidef}
\usepackage{gensymb}
\usepackage{xcolor}

\DeclareMathOperator*{\argmin}{argmin}

\newtheorem{thm}{Theorem}[]

\begin{document}
%

\title{User Association in Dense mmWave Networks based on Rate Requirements}


\author{
\IEEEauthorblockN{Veljko Boljanovic\IEEEauthorrefmark{1}, Forough Yaghoubi\IEEEauthorrefmark{2}, and Danijela Cabric\IEEEauthorrefmark{1}}
\IEEEauthorblockA{\IEEEauthorrefmark{1}Electrical and Computer Engineering Department, University of California, Los Angeles, CA, USA}
\IEEEauthorblockA{\IEEEauthorrefmark{2}Communication Systems Division, KTH Royal Institute of Technology, Stockholm, Sweden\\
Email: vboljanovic@ucla.edu, forough@kth.se, danijela@ee.ucla.edu}

\thanks{This work is supported by NSF under grant 1718742.}
}


\IEEEoverridecommandlockouts


\maketitle

\begin{abstract}
Commonly considered user association frameworks in millimeter-wave communications are based on the sum rate maximization, and they essentially neglect user specific rate and service requirements. Furthermore, new features of millimeter-wave communications including spatial multiplexing, connectivity to multiple coordinated base stations and dense base station deployment are not considered. In this work, we propose a two-step optimization framework for single-shot user association in dense millimeter-wave networks which takes into account users' rate requirements, multi-connectivity, and hybrid transceiver architecture for spatial multiplexing. The proposed framework considers multiple RF chains at each base station and assigns them to different users that also have multiple RF chains for connectivity with more than one base station. In the first step, the objective of the user association is to maximize the number of users with satisfied rate requirements while minimizing network underutilization. In the second step, remaining RF chains are assigned to users, whose rate requirement has not been met, such that network sum rate is maximized. This is a novel problem formulation for user associate in millimeter-wave networks. We propose low complexity sub-optimal user association algorithms based on this formulation, numerically evaluate the optimal and sub-optimal solutions, and compare them to the conventional association approaches in terms of the number of associated users and network sum rate.
\end{abstract}


%
\IEEEpeerreviewmaketitle
%
%
\section{Introduction}
\label{sec:introduction}
Millimeter-wave (mmWave) communications will have the key role in providing high data rates in the fifth generation (5G) of cellular systems \cite{Andrews:5G}. Besides the abundant spectrum at mmWave frequencies, spatial multiplexing enabled by hybrid analog-digital transceiver architectures with multiple RF chains will be leveraged for rate improvements and connectivity management. Hybrid antennas array architecture provides flexibility for data rate increase using simultaneous beamforming (BF) and spatial mutliplexing (SM) between the base station (BS) and user equipment (UE) \cite{Sun:multiplexVSbeamform}. Thus, in dense mmWave networks with large number of the UEs, the RF chains at the BS become an important resource for connectivity optimization. Due to its energy efficiency, it is expected that hybrid architecture will be implemented at UE side as well. Having multiple RF chains at the UE will allow the UE to connect to multiple BSs at the same time, which is often referred to as multi-connectivity \cite{3GPP:multi_connect}.  This new feature of multi-connectivity paired with high data rate requirements and limited number of RF chains per UE and BS makes the problems of user association (UA) and resource allocation (RA) in mmWave networks very challenging and different from microwave frequency approaches. 

The UA/RA at microwave frequencies were extensively studied for heterogeneous networks (HetNets), e.g., in \cite{Bethanabhotla:optimal, Liu:distributed, Ye:balancing} with the objective of maximizing network utility function. The utility function is often assumed to be logarithmic, which encourages load balancing in HetNets, while the set of constraints usually includes available resources at the BS or UE. In particular, the main Quality-of-Service (QoS) constraint in \cite{Bethanabhotla:optimal, Liu:distributed, Ye:balancing} is to maintain the minimum SNR or the maximum interference level, which may not necessarily lead to an acceptable rate for the UEs. A UE with a very high SNR associated with a heavily loaded BS may not get high effective rate due to potential sharing of the BS resources among many served UEs. More recently, UA/RA in mmWave networks were studied \cite{Ghadikolaei:assoc, Alizadeh:assoc, Tatino:multi_connect, Zhang:assoc}. Similarly to the microwave HetNets, the optimization frameworks in \cite{Ghadikolaei:assoc, Alizadeh:assoc, Tatino:multi_connect, Zhang:assoc} are focused on the sum rate maximization, while the sets of constraints mainly addressed resource availability and maximum interference levels. In \cite{Alizadeh:assoc}, UA/RA problem is considered for hybrid architecture with multiple RF chains, but it does not address multi-connectivity and users' data rate requirements. On the other hand, work \cite{Zhang:assoc} considers users' data rate requirements, but it does not study architectures with multiple RF chains and multi-connectivity. Finally, without considering users' data rate requirements, work \cite{Tatino:multi_connect} studies BSs and UEs with multiple RF chains and multi-connectivity, but it puts them in the context of time scheduling and handover in mmWave networks.

In this work, we propose an optimization framework for UA/RA in dense mmWave networks. By abstracting each RF chain at BS and UE as assignable system resource,  the UA/RA can be considered as the RF-chain-wise association between multiple BSs and UEs. Unlike previous works, our proposed framework jointly considers users' data rate requirements, and multi-connectivity. Moreover, our objective function is primarily to maximize the number of UEs with satisfied rate requirements, and then to assign the remaining BS RF chains to UEs such that network underutilization is avoided. We show that our optimization problem is NP-hard, and then develop a sub-optimal algorithm to solve the problem in polynomial time.

The rest of the paper is organized as follows. In Section~\ref{sec:models}, we introduce the system and channel models. In Section~\ref{sec:problem}, we introduce the proposed UA/RA optimization framework. Section~\ref{sec:proposed} describes the proposed sub-optimal algorithm. In Section~\ref{sec:evaluation}, we compare our framework with existing UA/RA approaches. Finally, conclusions are summarized in Section~\ref{sec:conclusions}.

%
%
\section{System and Channel Models}
\label{sec:models}

We consider downlink (DL) of a standalone mmWave network with a set of BSs $\mathcal{B}$ and a set of UEs $\mathcal{U}$, operating at frequency $f$. There are $N_{\text{BS}}$ BSs in $\mathcal{B}$, and $N_{\text{UE}}$ UEs in $\mathcal{U}$. We assume sub-array hybrid architecture with $N_{\text{BS}}^{\text{RF}}$ RF chains at each BS and $N_{\text{UE}}^{\text{RF}}$ RF chains at each UE. Each RF chain at each BS controls a uniform linear array (ULA) with $N_{\text{BS}}^{\text{a}}$ antenna elements, and each RF chain at each UE has a ULA with $N_{\text{UE}}^{\text{a}}$ antenna elements. We consider each RF chain at BS as a virtual BS and define a set of all RF chains at all BSs $\mathcal{B}_v$, whose cardinality is $N_{\text{BS}}N_{\text{BS}}^{\text{RF}}$. Similarly, we define a set of all RF chains at all UEs $\mathcal{U}_v$ with cardinality $N_{\text{UE}}N_{\text{UE}}^{\text{RF}}$. There are $N_{\text{BS}}N_{\text{BS}}^{\text{RF}}N_{\text{UE}}N_{\text{UE}}^{\text{RF}}$ multiple input multiple output (MIMO) channels between RF chains from $\mathcal{B}_v$ and $\mathcal{U}_v$.

Let $i$ and $j$ be arbitrary RF chains from $\mathcal{U}_v$ and $\mathcal{B}_v$, respectively. Using the bandwidth $B$, $i$ and $j$ communicate over a single-path mmWave MIMO channel represented by matrix $\mathbf{H}_{ij}\in\mathbb{C}^{N_{\text{UE}}^{\text{a}} \times N_{\text{BS}}^{\text{a}}}$. We consider dense urban micro environment and we model the channel path loss\footnote{All RF chain pairs $(i,j)$ that correspond to the same BS-UE pair experience the same path loss because of spatial consistency.} according to \cite{3GPP:channel}. Under assumption that capacity achieving code is used, achievable data rate between $i$ and $j$ can be approximated with the link capacity. Since highly directional transmission in mmWave networks is noise-limited rather than interference-limited \cite{Sun:multiplexVSbeamform}, the link capacity $c_{ij}$ is calculated as follows
\begin{align}
    c_{ij} = B\log_2\left( 1 + \frac{|\frac{\sqrt{P_t}}{N_{\text{BS}}^{\text{RF}}} \mathbf{a}_{\text{UE}}^{\text{H}}(\hat{\theta}_i) \mathbf{H}_{ij} \mathbf{a}_{\text{BS}}(\hat{\phi}_j)|^2}{B N_0} \right),
    \label{eq:capacity}
\end{align}
where $P_t$ and $N_0$ represent the transmit power and noise spectral density, respectively. The beamforming vectors are equal to the spatial response vectors $\mathbf{a}_{\text{UE}}(\hat{\theta}_i)$ and $\mathbf{a}_{\text{BS}}(\hat{\phi}_j)$ defined as follows
\begin{align}
    \mathbf{a}_{\text{UE}}(\hat{\theta}_i) = \frac{[1,e^{-j\pi\sin(\hat{\theta}_i)},...,e^{-j(N_{\text{UE}}^{\text{a}}-1)\pi\sin(\hat{\theta}_i)}]^{\text{T}}}{\sqrt{N_{\text{UE}}^{\text{a}}}},
    \label{eq:spatialAoA}
\end{align}
\begin{align}
    \mathbf{a}_{\text{BS}}(\hat{\phi}_j) = \frac{[1,e^{-j\pi\sin(\hat{\phi}_j)},...,e^{-j(N_{\text{BS}}^{\text{a}}-1)\pi\sin(\hat{\phi}_j)}]^{\text{T}}}{\sqrt{N_{\text{BS}}^{\text{a}}}}.
    \label{eq:spatialAoD}
\end{align}
The imperfect angular estimates $\hat{\theta}_i$ and $\hat{\phi}_j$ are assumed to be obtained through practical beam training. They are modeled as Gaussian random variables $\hat{\theta}_{ij}\sim\mathcal{N}(\theta_{ij}, \sigma^2_{\text{AoA}})$ and $\hat{\phi}_{ij}\sim\mathcal{N}(\phi_{ij}, \sigma^2_{\text{AoD}})$, where $\theta_{ij}$ and $\phi_{ij}$ represent true angle of arrival (AoA) and angle of departure (AoD), respectively. The imperfect estimates $\hat{\theta}_i$ and $\hat{\phi}_j$ can negatively affect the capacity (achievable rate) between $i$ and $j$ due to decrease in beamforming gain.

%
%
\section{Optimization Framework}
\label{sec:problem}
In an urban dense environment, 3GPP specifices that each UE should get data rate of at least $300$ Mbps \cite{3GPP:system_req}. Note that UEs’ data requirements can be significantly higher and reach the order of several Gbps. Conventional UA/RA schemes often do not consider heterogeneous UEs’ data requirements, and resources are often allocated to UEs with high capacity links. In this work, we take data requirements into account, and we design a new two-step optimization framework for UA/RA, assuming that RF chains at BSs represent the resources. In the first step, we maximize the number of associated UEs with satisfied data requirements. This maximization is done with minimal number of BS RF chains to avoid network underutilization. Simultaneously, RF chains are allocated such that the sum rate among associated UEs is maximal. Since minimal amount of resources is used in the first step, in the second step the remaining RF chains are used to serve non-associated users and maximize network sum rate. 


\subsection{Step 1}
\label{sec:QoS}
Let $\mathbf{z}\in \{0,1\}^{|\mathcal{U}|}$ be a vector of binary association variables, where $z_u$ is $1$ if user $u$ is served and $0$ otherwise. Let $\mathbf{r}$ be a vector of data rate requirements $r_u$ for all $u\in\mathcal{U}$. Let $\mathbf{x} \in \{0,1\}^{|\mathcal{U}_v||\mathcal{B}_v|}$ be a vectorized matrix of binary association variable for all RF chain pairs, where $x_{ij}$ is $1$ if user RF chain $i$ is connected to BS RF chain $j$ and $0$ otherwise. Let $\mathbf{c}$ be a vector whose elements are capacities $c_{ij}$ from (\ref{eq:capacity}), associated with corresponding $x_{ij}$. We define the function $F_1$ as the number of associated UEs, and $F_2$ as the number of allocated BS RF chains with negative sign. Mathematically, $F_1$ and $F_2$ can be expressed as follows
\begin{align}
    F_1 = \sum_{u\in\mathcal{U}}z_u,~~~~~ F_2 = -\sum_{i\in\mathcal{U}_v} \sum_{j\in\mathcal{B}_v} x_{ij}.
\end{align}

The goal in Step 1 is to maximize $F_1$ with maximal $F_2$. To achieve this, a multi-criterion optimization problem needs to be solved. Since the problem is constrained on users' data requirements and available resources at BSs, it can be formulated as follows
\begin{maxi!}|s|[2]{\mathbf{z}, \mathbf{x}}{\lambda_1 F_1 + \lambda_2 F_2} {\label{eq:opt1}}{}
    \addConstraint{\sum_{i\in\mathcal{U}_v}x_{ij}}{\leq 1,~~~}{\forall j \in \mathcal{B}_v}{\label{eq:opt1_const_1}}
    \addConstraint{\sum_{j\in\mathcal{B}_v}x_{ij}}{\leq 1,~~~}{\forall i \in \mathcal{U}_v}{\label{eq:opt1_const_2}}
    \addConstraint{\sum_{i \in \mathcal{U}_v}\sum\limits_{\substack{j\in\mathcal{B}_v \\ j \rightarrow b}}x_{ij}}{\leq N_{\text{BS}}^{\text{RF}},~~~}{\forall b \in \mathcal{B}}{\label{eq:opt1_const_3}}
    \addConstraint{\sum\limits_{\substack{i\in\mathcal{U}_v \\ i \rightarrow u}}\sum_{{j\in\mathcal{B}_v}}x_{ij}}{\leq z_u N_{\text{UE}}^{\text{RF}},~~~}{\forall u \in \mathcal{U}}{\label{eq:opt1_const_4}}
    \addConstraint{\sum\limits_{\substack{i\in\mathcal{U}_v \\ i \rightarrow u}}\sum_{{j\in\mathcal{B}_v}}x_{ij}c_{ij}}{\geq z_u r_u,~~~}{\forall u \in \mathcal{U}}{\label{eq:opt1_const_5}}
    \addConstraint{x_{ij}}{\in\{0,1\},~~~}{\forall i \in \mathcal{U}_v, j \in \mathcal{B}_v}{\label{eq:opt1_const_6}}
    \addConstraint{z_{u}}{\in\{0,1\},~~~}{\forall u \in \mathcal{U},}{\label{eq:opt1_const_7}}
\end{maxi!}
where $j \rightarrow b$ means that the RF chain $j \in \mathcal{B}_v$ belongs to the BS $b \in \mathcal{B}$. Similarly,  $i \rightarrow u$ means that the RF chain $i \in \mathcal{U}_v$ belongs to the UE $u \in \mathcal{U}$. Constraints (\ref{eq:opt1_const_1}) and (\ref{eq:opt1_const_2}) guarantee that each RF chain can be connected to up to 1 RF chain. Constraints (\ref{eq:opt1_const_3}) and (\ref{eq:opt1_const_4}) relate to the maximum number of RF chains at the BS $b$ and UE $u$, respectively. The variable $z_u$ in (\ref{eq:opt1_const_4}) ensures that RF chains of UE $u$ are not used if $u$ is not associated. The data requirement constraint in (\ref{eq:opt1_const_5}) guarantees that all associated UEs have their data rate requirements satisfied. The constants $\lambda_1$ and $\lambda_2$ in (\ref{eq:opt1}a) represent weights which can be obtained through scalarization. The set of all possible values for $F_1$ and $F_2$ includes the optimal trade-off curve that is a piece-wise linear (PWL) function consisting of discrete points. It is possible to find a hyperplane defined by $[\lambda_1, \lambda_2]$ which touches the optimal trade-off curve at the point where $F_1$ is maximized and $F_2$ is maximal. Commonly, the weight $\lambda_1$ is fixed, and then the optimal values for $\lambda_2$ are found. Based on capacities of its links, a UE could need from $1$ to $N_{\text{UE}}^{\text{RF}}$ RF chains to satisfy its data rate requirement. This range defines a set of slopes of the optimal trade-off curve. They could take values from the set $\{1,\frac{1}{2},...,\frac{1}{N_{\text{UE}}^{\text{RF}}}, 0\}$. The optimal values for $\lambda_2$ directly depend on the slopes. An example for finding $\lambda_2$ when $\lambda_1=1$ and $N_{\text{UE}}^{\text{RF}}=2$ is depicted in Fig.~\ref{fig:pareto_optimal}. If $\lambda_2=1$, the number of associated users, i.e., objective $F_1$, is not maximized and there are multiple optimal points. If $\lambda_2=\frac{1}{N_{\text{UE}}^{\text{RF}}}$, there are again multiple optimal points and $F_1$ is not necessarily maximized. Further, if $\lambda_2=0$, the objective $F_1$ is maximized, but the number of used RF chains is not minimal. It can be observed that $F_1$ is maximized with minimal number of RF chains (the red point is certainly achieved) if ${\lambda_2\in \left(0,\frac{1}{N_{\text{UE}}^{\text{RF}}} \right)}$. Note that this result holds for any $N_{\text{UE}}^{\text{RF}}$ and any optimal trade-off curve. In this work, we choose ${\lambda_2 = \frac{1}{N_{\text{UE}}^{\text{RF}}+1}}$, and then ($\ref{eq:opt1}$a) gets the following form
\begin{align}
    \max_{\mathbf{z},  \mathbf{x}}~~\sum_{u\in\mathcal{U}}z_u - \sum_{i\in\mathcal{U}_v} \sum_{j\in\mathcal{B}_v} \frac{1}{N_{\text{UE}}^{\text{RF}}+1} x_{ij}.
    \label{eq:random_opt1}
\end{align}

Note that with $[\lambda_1, \lambda_2]=\left[1,\frac{1}{{N_{\text{UE}}^{\text{RF}}+1}} \right]$, the sum rate among associated UEs is not necessarily maximal since they are associated using \textit{arbitrary} links that satisfy their data requirements and maximize (\ref{eq:random_opt1}). We further extend $\lambda_2$ by adding the term $\frac{c_{ij}}{r_u}$ to its denominator and reformulate (\ref{eq:random_opt1}) as follows
\begin{align}
    \max_{\mathbf{z},  \mathbf{x}}~~\sum_{u\in\mathcal{U}}z_u - \sum_{u\in\mathcal{U}} \sum\limits_{\substack{i\in\mathcal{U}_v \\ i \rightarrow u}} \sum_{j\in\mathcal{B}_v} \frac{1}{N_{\text{UE}}^{\text{RF}}+1+\frac{c_{ij}}{r_u}} x_{ij}.
    \label{eq:extended}
\end{align}
The last expression can be considered as ($\ref{eq:opt1}\text{a}$) with a new pair of functions $F_1^{\prime}$ and $F_2^{\prime}$ and ${[\lambda_1^{\prime}, \lambda_2^{\prime}]=[1,1]}$. The formulation in (\ref{eq:extended}) ensures that the sum rate of associated UEs is maximal. The objective function increases more if UEs are associated with links that have higher capacity $c_{ij}$. If capacity $c_{ij}$ was not divided by corresponding $r_u$, minimal number of used RF chains would not be guaranteed. To see this, assume $u_1$ and $u_2$ are two UEs from $\mathcal{U}$, where $u_1$ experience only low capacity links with all BSs, but it can satisfy its low data rate requirement with one RF chain, and $u_2$ that has high capacity links with all BSs, but it needs two RF chains to satisfy its extremely high rate requirement. If capacity $c_{ij}$ was not divided by corresponding $r_u$, high capacity user $u_2$ would be favored even though it requires more RF chains. In other words, the objective could see higher reward in associating $u_2$ than $u_1$. To solve this issue, we introduce the relative capacity term $\frac{c_{ij}}{r_u}$ which ensures that the sum rate is maximal among UEs associated with minimal number of RF chains. With the relative capacity term, the objective function sees higher reward if UEs are associated with smaller number of RF chains. 
Finally, the optimization problem in Step 1 can be restated as follows
\begin{figure}
    \begin{center}
        \includegraphics[width=0.48\textwidth]{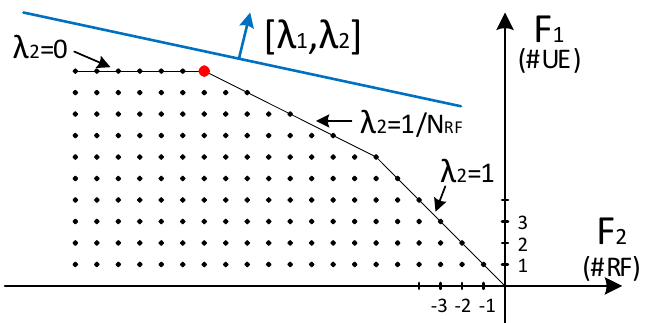}
    \end{center}
    \vspace{-4mm}
    \caption{PWL trade-off curve with discrete points. We fix $\lambda_1=1$ and find the values of $\lambda_2$ for which the red optimal point is certainly achieved.}
    \vspace{-4mm}
    \label{fig:pareto_optimal}
\end{figure}
\begin{maxi}|s|[2]{\mathbf{z}, \mathbf{x}}{\sum_{u\in\mathcal{U}}z_u - \sum_{u\in\mathcal{U}} \sum\limits_{\substack{i\in\mathcal{U}_v \\ i \rightarrow u}} \sum_{j\in\mathcal{B}_v} \frac{1}{N_{\text{UE}}^{\text{RF}}+1+\frac{c_{ij}}{r_u}} x_{ij}} {\label{eq:opt2}}{}
    \addConstraint{\text{(\ref{eq:opt1_const_1})-(\ref{eq:opt1_const_7})}}{}{}
\end{maxi}

In summary, Step 1 in the optimization framework maximizes the number of associated UEs with satisfied data rate requirements using minimal number of RF chains at BSs. In addition, it ensures maximal sum rate among associated UEs when minimal resources are used.

\subsection{Step 2}
After Step 1, some number of RF chains at BSs might still be available. Since no more UEs can satisfy their data requirements, the reminder of resources at BSs is fully exploited and distributed among non-associated UEs using max-sum-rate scheme. Note that the proposed framework reduces to the sum rate maximization among all users if no users are associated in Step 1.

Let $\mathcal{U}=\mathcal{U}^{\text{A}} \cup \mathcal{U}^{\text{NA}}$, where $\mathcal{U}^{\text{A}}$ and $\mathcal{U}^{\text{NA}}$ are sets of associated and non-associated UEs in Step 1, respectively. Now let $\mathcal{U}_v^{\text{NA}}$ be a set of all RF chains at non-associated UEs. Let $\mathcal{B}_v=\mathcal{B}_v^{\text{A}} \cup \mathcal{B}_v^{\text{NA}}$, where $\mathcal{B}_v^{\text{A}}$ and $\mathcal{B}_v^{\text{NA}}$ are sets of assigned and non-assigned RF chains from BSs in Step 1, respectively. Let $N_{\text{BS}(b)}^{\text{RF}}$ be the number of remaining RF chains at the BS $b\in\mathcal{B}$ after Step 1. Similarly as in Step 1, let $\mathbf{x}_s$ be a vector of binary association variables $x_{ij}$ for all RF chain pairs $\{i,j\},~i \in \mathcal{U}_v^{\text{NA}}, j \in \mathcal{B}_v^{\text{NA}}$. Note that the vector $\mathbf{x}_s$ consists of the subset of variables from the vector $\mathbf{x}$. Let $c_{ij}$ be the capacity associated with corresponding variable $x_{ij}$ from $\mathbf{x}_s$. The sum rate can be maximized as follows
\begin{maxi!}|s|[2]{\mathbf{x}_s}{\sum_{i\in\mathcal{U}_v^{\text{NA}}} \sum_{j\in\mathcal{B}_v^{\text{NA}}} x_{ij}c_{ij}} {\label{eq:sum_rate}}{}
    \addConstraint{\sum_{i\in\mathcal{U}_v^{\text{NA}}}x_{ij}}{\leq 1,~~~}{\forall j \in \mathcal{B}_v^{\text{NA}}}{\label{eq:sum_rate_const_1}}
    \addConstraint{\sum_{j\in\mathcal{B}_v^{\text{NA}}}x_{ij}}{\leq 1,~~~}{\forall i \in \mathcal{U}_v^{\text{NA}}}{\label{eq:sum_rate_const_2}}
    \addConstraint{\sum_{i \in \mathcal{U}_v^{\text{NA}}}\sum\limits_{\substack{j\in\mathcal{B}_v^{\text{NA}} \\ j \rightarrow b}}x_{ij}}{\leq N_{\text{BS}(b)}^{\text{RF}},~~~}{\forall b \in \mathcal{B}}{\label{eq:sum_rate_const_3}}
    \addConstraint{\sum\limits_{\substack{i\in\mathcal{U}_v^{\text{NA}} \\ i \rightarrow u}}\sum_{{j\in\mathcal{B}_v^{\text{NA}}}}x_{ij}}{\leq N_{\text{UE}}^{\text{RF}},~~~}{\forall u \in \mathcal{U}^{\text{NA}}}{\label{eq:sum_rate_const_4}}
    \addConstraint{x_{ij}}{\in\{0,1\},~~~}{\forall i \in \mathcal{U}_v^{\text{NA}}, j \in \mathcal{B}_v^{\text{NA}}}{\label{eq:sum_rate_const_5}}.
\end{maxi!}
As in Step 1, constraints (\ref{eq:sum_rate_const_1}) and (\ref{eq:sum_rate_const_2}) guarantee that each RF chain can be connected to up to 1 RF chain. Constraints (\ref{eq:sum_rate_const_3}) and (\ref{eq:sum_rate_const_4}) ensure that the number of used RF chains does not exceed the amount of available resources at the BS and non-associated UE, respectively. The formulation in (\ref{eq:sum_rate}) maximizes the sum rate in Step 2 regardless of user data rate requirements, but other approaches are also possible. For example, the association scheme in Step 2 can be designed to distribute remaining resources among non-associated UEs according to their data rate requirements. With this approach, UEs with higher data requirements would get more resources. We leave these alternative designs for future work.

%
%
\section{Proposed Algorithm}
\label{sec:proposed}
The optimization problems in (\ref{eq:opt2}) and (\ref{eq:sum_rate}) are Binary Integer Programs (BIP), which are known to be NP-hard. This means that even small-size problems with a few BSs and moderate number of UEs have prohibitive computational complexity. To solve these optimization problems, we propose a low complexity solution based on relaxation and rounding. First, we relax the BIP in (\ref{eq:opt2}) and formulate the following low complexity Linear Program (LP)
\begin{maxi}|s|[2]{\mathbf{z}, \mathbf{x}}{\sum_{u\in\mathcal{U}}z_u - \sum_{u\in\mathcal{U}} \sum\limits_{\substack{i\in\mathcal{U}_v \\ i \rightarrow u}} \sum_{j\in\mathcal{B}_v} \frac{1}{N_{\text{UE}}^{\text{RF}}+1+\frac{c_{ij}}{r_u}} x_{ij}} {\label{eq:opt2_relax}}{}
    \addConstraint{\text{(\ref{eq:opt1_const_1})-(\ref{eq:opt1_const_5})}}{}
    \addConstraint{0 \leq x_{ij}}{\leq 1,~~~}{\forall i \in \mathcal{U}_v, j \in \mathcal{B}_v}{}
    \addConstraint{0 \leq z_{u}}{\leq 1,~~~}{\forall u \in \mathcal{U},}{}
\end{maxi}
\begin{algorithm}
\caption{Proposed rounding algorithm}
\label{algorithm}
\begin{algorithmic}[1]
\State $\textbf{Inputs:} ~~\mathbf{C}_{\mathcal{S}}, ~\mathbf{r}$ 
\State $\textbf{Outputs:} ~~\mathbf{z}_{ro}, ~\mathbf{x}_{ro}$
\State $\textbf{Initialization:} ~~\mathbf{z}_{ro}=\mathbf{0}, ~\mathbf{x}_{ro}=\mathbf{0}$
\State $\mathbf{X}_{ro} = \text{vec}^{-1}(\mathbf{x}_{ro})$
\State Define matrix of values $\mathbf{V}$ and matrix of indices $\mathbf{I}$
\For {$n=1:N_{\text{UE}}^{\text{RF}}$}
    \While {$true$}
        \State $\mathbf{d}=\mathbf{0}$
        \State $[\mathbf{V}[:,u],~\mathbf{I}[:,u]] = \text{sort} \left( \mathbf{C}_{\mathcal{S}} [:,u],~'\text{descend}' \right)$
        \State $d_u = \argmin_k \left( \sum_{m=1}^{k} \mathbf{V} [m,u] \geq r_u \right), \forall u$
        \If {$d_u \neq n,~\forall u$}
            \State \textbf{break while}
        \EndIf
        \State $\mathbf{z}_{ro}[u^*]=1$, for $[\sim,u^*] = \max \left( \sum_{m=1}^{n} \mathbf{V} [m,:] \right)$
        \State $\mathbf{X}_{ro}[\mathbf{I}[1:n,u^*], u^*]=1$
        \State $\mathbf{C}_{\mathcal{S}} [:,u^*] = \mathbf{0}$
        \State $\mathbf{C}_{\mathcal{S}} [k,:] = \mathbf{0}^{\text{T}}$, $\forall k$ of assigned BS RF chains
    \EndWhile
\EndFor
\State $\mathbf{x}_{ro} = \text{vec}\left(\mathbf{X}_{ro} \right)$
\end{algorithmic}
\end{algorithm}

The solution $\mathbf{x}^*$ to (\ref{eq:opt2_relax}) is fractional, meaning that its elements are not necessarily integer values $0$ or $1$. Since we consider a single-shot UA and not the long-term average, elements in $\mathbf{x}^*$ should be rounded either to $0$ or $1$. Using simple rounding technique where values $x_{ij}^* \geq 0.5$ are rounded to $1$, and zero otherwise, is not a good way to find the rounded solution $\mathbf{x}_{ro}$. This technique can round too many elements to $1$, and thus $\mathbf{x}_{ro}$ can violate multiple constraints and become infeasible. Similarly, it can round too many elements to $0$, and then $\mathbf{x}_{ro}$ becomes a poor solution with few associated UEs. We propose a sub-optimal polynomial time rounding algorithm used to obtain $\mathbf{x}_{ro}$ from $\mathbf{x}^*$.

Let $\mathcal{S}$ be the support of the vector $\mathbf{x}^*$. Let vector $\mathbf{c}_{\mathcal{S}}$ be equal to the vector $\mathbf{c}$ for indices found in $\mathcal{S}$, and zero for all other indices. Let $\mathbf{X}^*=\text{vec}^{-1}(\mathbf{x}^*)$, where $\text{vec}^{-1}()$ reshapes a vector of size $\mathbb{R}^{N_{\text{BS}}N_{\text{BS}}^{\text{RF}}N_{\text{UE}}N_{\text{UE}}^{\text{RF}}}$ into a matrix of size $\mathbb{R}^{N_{\text{BS}}N_{\text{BS}}^{\text{RF}}N_{\text{UE}}^{\text{RF}}\times N_{\text{UE}}}$. Similarly, let $\mathbf{C}_{\mathcal{S}}=\text{vec}^{-1}(\mathbf{c}_{\mathcal{S}})$ be a matrix of capacities on used links. Columns in $\mathbf{X}^*$  and $\mathbf{C}_{\mathcal{S}}$ correspond to different UEs, and rows correspond to all possible links that UEs can have.

We first compare elements $\mathbf{C}_{\mathcal{S}}$ to the corresponding data rate requirements $r_u,~u=1,...,N_{\text{UE}}$, to see how many RF chains each associated UE needs. The needed number of RF chains is stored in the demand vector $\mathbf{d}$. We identify all UEs that can be associated using 1 RF chain, i.e., we find all positions in $\mathbf{d}$ where $d_u=1$. Among these UEs, we find the one that has link with the highest capacity and associate the UE using this link. To exclude the associated UE from further consideration, its corresponding column in $\mathbf{C}_{\mathcal{S}}$ is set to $\mathbf{0}$. Similarly, we exclude the used BS RF chain from further consideration by setting all corresponding rows in $\mathbf{C}_{\mathcal{S}}$ to $\mathbf{0}^{\text{T}}$. In the next iteration, the demand vector $\mathbf{d}$ is determined again, and the whole procedure is repeated for $d_u=1$ if there are still UEs that require 1 RF chain. If that is not the case, the procedure first repeats for $d_u=2$, then for $d_u=3$, and so on until $d_u=N_{\text{UE}}^{\text{RF}}$. Note that when $d_u>1$, we pick the UE with the highest aggregated capacity on $d_u$ best links. The algorithm pseudo-code is provided in Algorithm \ref{algorithm}.

The sub-optimality of the proposed algorithm comes from the fact that it does not maximize the sum rate of associated UEs. In fact, this is a greedy approach which tries to maximize the number of associated UEs with minimal number of RF chains by choosing the UE with best links in each iteration. The greedy selection of best links often comes at price of lower number of associated UEs. Consequently, the proposed solution leaves more RF chains for the second step, where the sum rate is maximized.

Once vectors $\mathbf{z}_{ro}$ and $\mathbf{x}_{ro}$ are obtained, the optimal $\mathbf{x}_s^*$ in Step 2 can be obtained by relaxing and solving (\ref{eq:sum_rate}).
\begin{thm}
The relaxation of (\ref{eq:sum_rate}) has an integral optimal solution $\mathbf{x}_s^*$, whose elements are 0 or 1.
\label{th:1}
\end{thm}
\begin{proof}
See Appendix~\ref{app:integer}.
\end{proof}

\vspace{1mm}
\section{Numerical Evaluation}
\label{sec:evaluation}
In this section, we evaluate the proposed low complexity algorithm for new UA/RA framework, and compare it to existing association schemes, including the max-sum-rate and max-SNR. Neiher the max-sum-rate nor max-SNR consider users' data requirements and they both allocate BS RF chains to the UEs with high capacity links. The max-sum-rate jointly considers all RF chains at all BSs to maximize the network sum rate, while the max-SNR considers resources at one BS at a time and thus allocates them sequentially.

We consider a scenario with $N_{\text{BS}}=5$ BSs and $N_{\text{UE}}=30$ UEs randomly placed within the area among BSs. The distance between neighboring BSs is $200\text{m}$. We assume hybrid sub-array architecture at both BSs and UEs, with $N_{\text{BS}}^{\text{RF}}=5$ RF chains at each BS and $N_{\text{UE}}^{\text{RF}}=2$ RF chains at each UE. Each BS RF chain controls $N_{\text{BS}}^{\text{a}}=32$ antennas, while each UE RF chain controls $N_{\text{UE}}=8$ antennas. The data rate requirements are drawn from uniform distribution $r_u\sim U(R_{\text{min}}, R_{\text{max}}),~\forall u$, where $R_{\text{min}}=0.3~\text{Gbps}$ and $R_{\text{max}}$ can vary. We consider the downlink communication with BS transmit power $P_t=30~\text{dBm}$, and noise spectral density of $N_0=-174~\frac{\text{dBm}}{\text{Hz}}$. All RF chain pairs $\{i,j\},~i\in\mathcal{U}_v, j\in\mathcal{B}_v$, operate over the same bandwidth of $B=200~\text{MHz}$. True angles $\theta_{i}$ and $\phi_{j}$ are drawn from uniform distribution $U(-\frac{\pi}{2},\frac{\pi}{2})$ for each RF chain pair $(i,j)$, and standard deviations of estimation errors are $\sigma_{\text{AoD}}=1^{\circ}$ and $\sigma_{\text{AoA}}=3^{\circ}$ for AoD and AoA, respectively.
\begin{figure}[t]
\begin{tabular}{cc}
\vspace{1mm}
\subfloat[User association based on optimal solutions to (\ref{eq:opt2}) and (\ref{eq:sum_rate}).]{%
  \includegraphics[clip,width=0.95\columnwidth]{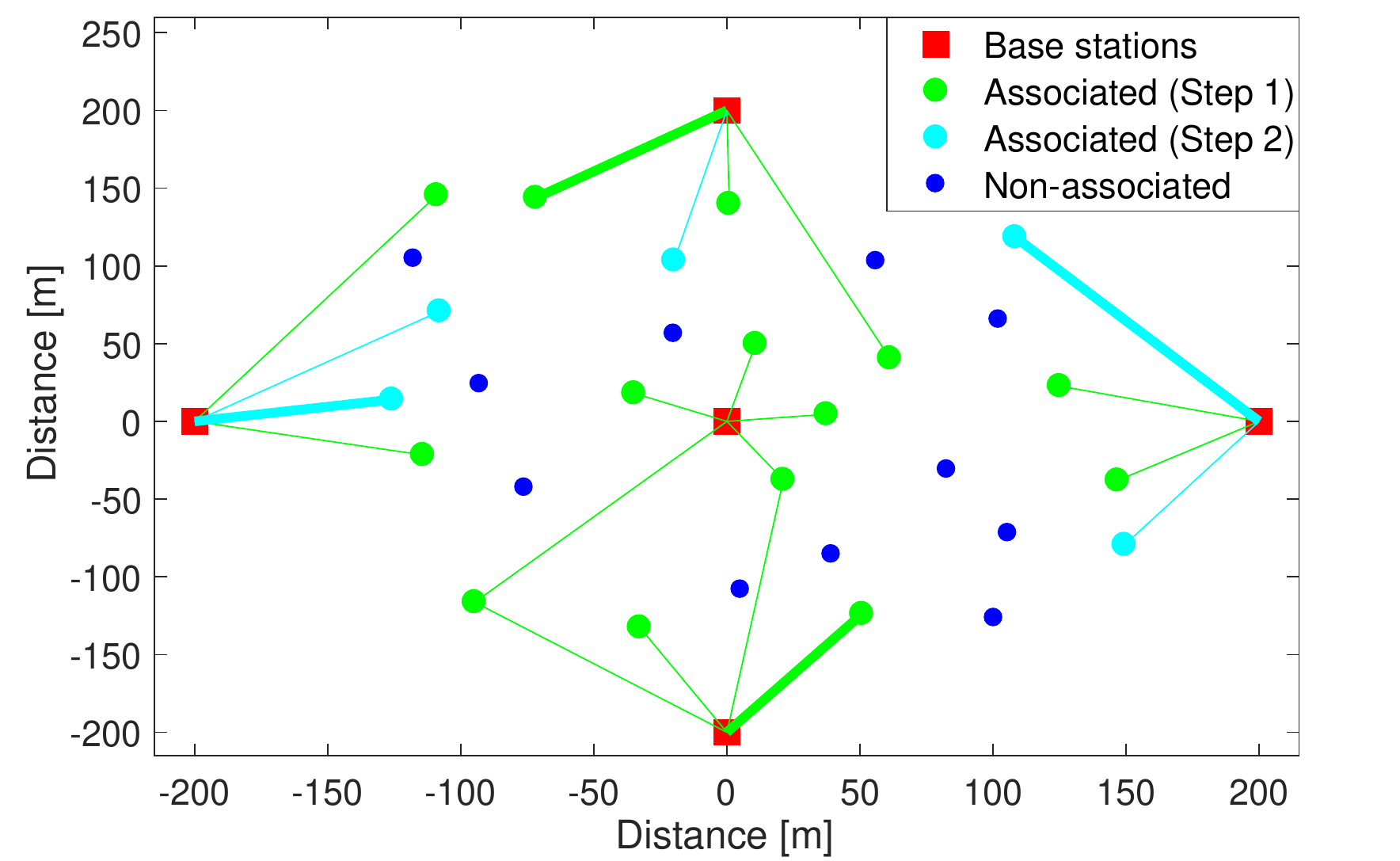}%
}\\
\vspace{-3mm}
\subfloat[User association based on proposed solutions to (\ref{eq:opt2}) and (\ref{eq:sum_rate}).]{%
  \includegraphics[clip,width=0.95\columnwidth]{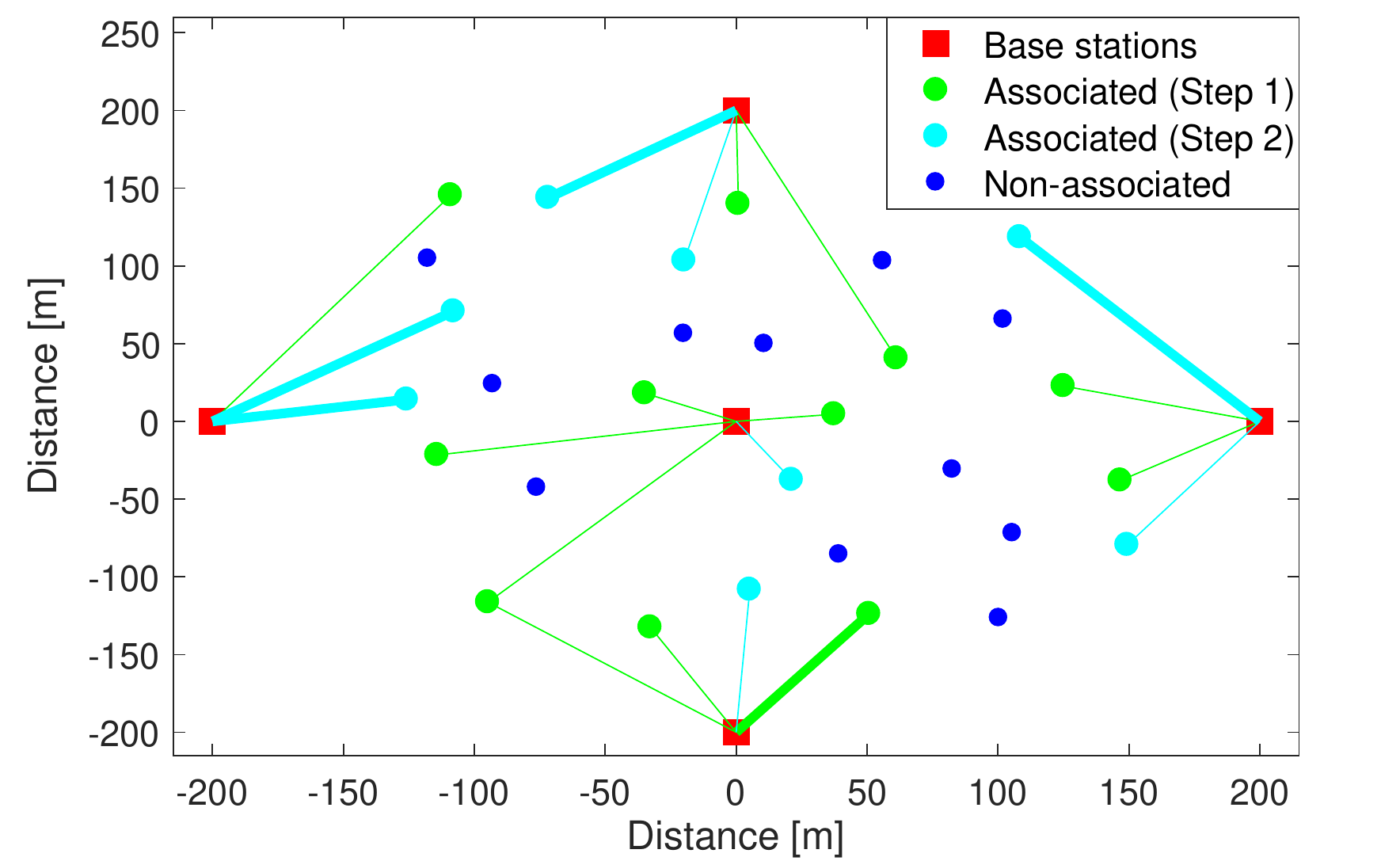}%
}\\
\vspace{0mm}
\end{tabular}
\caption{User association for one scenario realization. Thin and thick lines illustrate the use of 1 and 2 RF chains for communication, respectively.}
\vspace{-5mm}
\label{fig:opt_vs_proposed}
\end{figure}

The UA/RA in the proposed two-step optimization framework is presented in Fig.~\ref{fig:opt_vs_proposed} for one scenario realization, with ${R_{\text{max}}=2~\text{Gbps}}$. The two subfigures compare the associations based on the optimal and proposed solutions to (\ref{eq:opt2}) and (\ref{eq:sum_rate}). Association based on the optimal solutions maximizes the number of UEs with satisfied data rate requirement in Step 1, and it leaves small number of RF chains for Step 2. On the other hand, the UA based on the proposed sub-optimal solutions associates less users with satisfied data requirements in Step 1 due to its greedy nature. This means that more RF chains are available in Step 2, where they are given to non-associated UEs through the max-sum-rate approach formulated in (\ref{eq:sum_rate}). When hybrid transceiver architecture is considered, a BS tends to allocate multiple RF chains for a single UE with good links in Step 2. This explains why many UEs are provided with two RF chains in Step 2 in Fig.~\ref{fig:opt_vs_proposed}(b).

In Fig.~\ref{fig:comparison}, the proposed UA/RA framework is compared with existing UA schemes, including max-sum-rate and max-SNR, for the same scenario realization as in Fig.~\ref{fig:opt_vs_proposed}. As expected, the UA/RA based on the optimal solutions to (\ref{eq:opt2}) and (\ref{eq:sum_rate}) results in the highest number of associated UEs with satisfied data rate requirements. The proposed sub-optimal UA/RA performs better than conventional association schemes in terms of the number of associated users. When conventional schemes are used for association, only UEs with good links are provided with the opportunity to satisfy their data rate requirements regardless of how low or high their requirements are. Unsurprisingly, the max-sum-rate scheme achieves the highest network sum rate. The sum rate in the proposed framework is higher with the sub-optimal UA/RA than with the optimal. The main reasons for this are sub-optimal maximization with the proposed solutions in Step 1 and consequent use of more RF chains in Step 2 where sum rate is maximized among non-associated UEs.
\begin{figure}
    \begin{center}
        \includegraphics[width=0.48\textwidth]{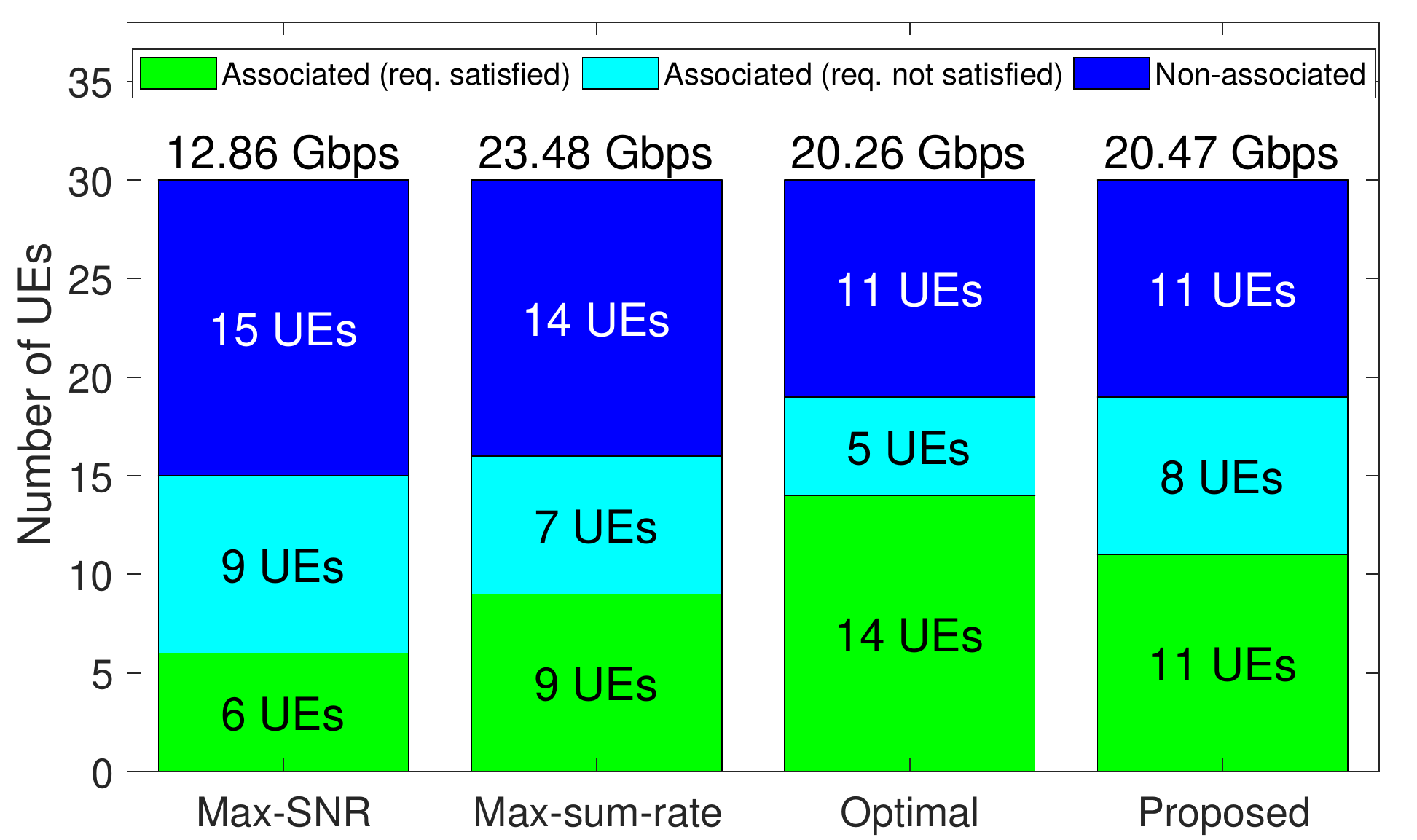}
    \end{center}
    \vspace{-4mm}
    \caption{Comparison in terms of the number of associated UEs and achieved network sum rate for one scenario realization. The associated UEs can have their data requirements satisfied or not.}
    \vspace{-6mm}
    \label{fig:comparison}
\end{figure}

In Fig.~\ref{fig:mc_comparison}, the average number of associated UEs with satisfied data rate requirements and the average achieved sum rate are presented as functions of the maximum requirement $R_{\text{max}}$. To obtain results in both subfigures, we perform 30 Monte Carlo runs with different positions of UEs and users' data requirements $r_u,~\forall u$, to find the averages for different $R_{\text{max}}$. By jointly considering both subfigures, we see that the UA/RA based on the proposed solutions represents a trade-off between max-sum-rate scheme and UA/RA based on the optimal solutions. When the proposed solutions are used, the number of associated UEs with satisfied data requirements is higher than with the max-sum-rate scheme, and lower than with optimal solutions. On the other hand, the use of the proposed solutions results in higher sum rate than with the optimal solutions, but lower than with the max-sum-rate scheme. As the number of associated UEs in Step 1 of the proposed framework decreases in Fig.~\ref{fig:mc_comparison}(a), the sum rate increases in Fig.~\ref{fig:mc_comparison}(b) since more BS RF chains are available for rate maximization in Step 2. Note that the sum rates for the max-sum-rate and max-SNR schemes in Fig.~\ref{fig:mc_comparison}(b) do not increase because these schemes associate UEs with good links in a single step, regardless of UEs' data requirements.

%
%
\section{Conclusions}
\label{sec:conclusions}
We proposed a new UA/RA framework which, unlike existing UA/RA approaches, considers features of mmWave networks and diverse user data rate requirements. The first step in the framework maximizes the number of associated UEs with satisfied data rate requirements using minimal amount of resources, while the second step uses the remaining resources to maximize the sum rate among non-associated UEs. We proposed sub-optimal solutions to the NP-hard problems in both steps, and our numerical results showed that the proposed solutions represent a good trade-off between the optimal solutions to the NP-hard problems and existing max-sum-rate approach.
\begin{figure}[t]
\begin{tabular}{cc}
\vspace{1mm}
\subfloat[Average number of associated UEs with satisfied QoS requirements.]{%
  \includegraphics[clip,width=0.95\columnwidth]{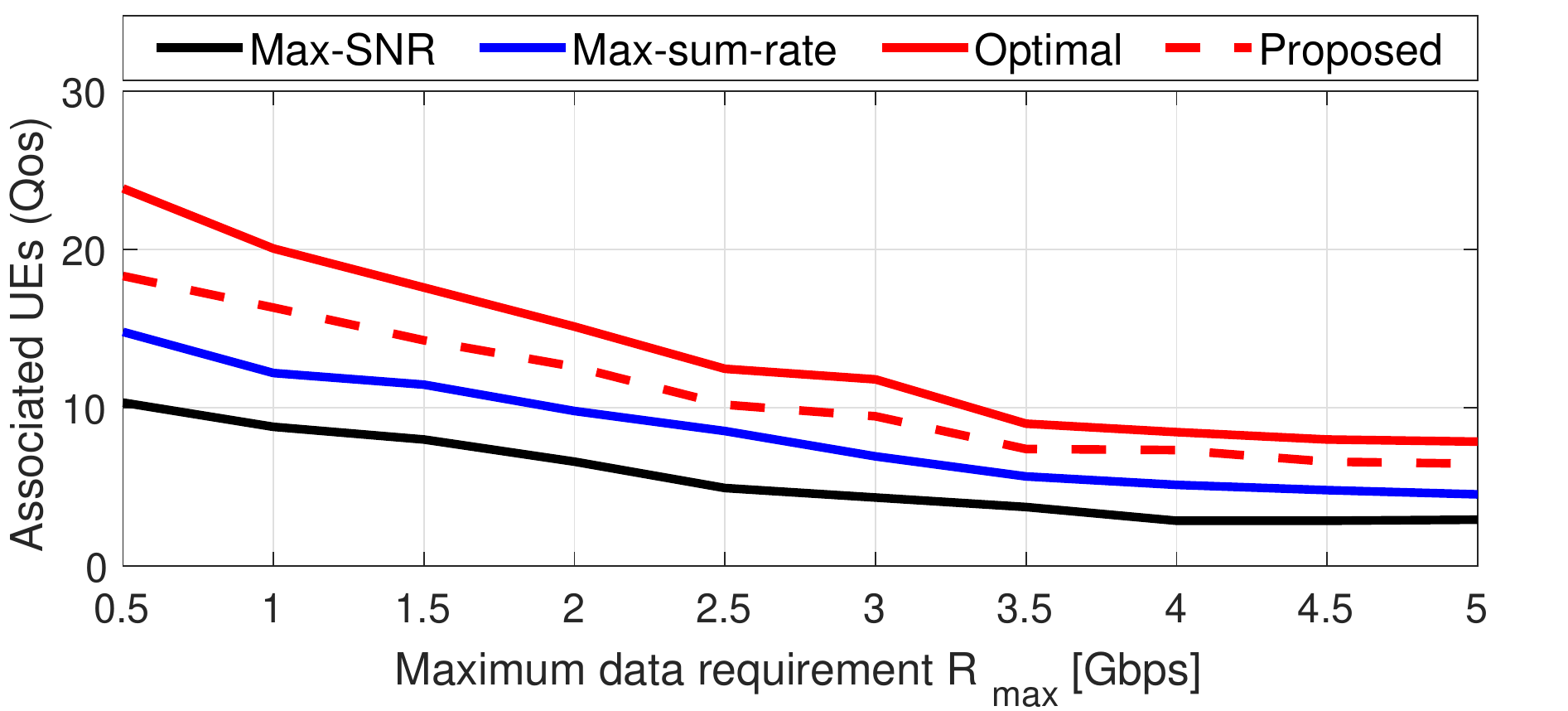}%
}\\
\vspace{-3mm}
\subfloat[Average sum rate in the network.]{%
  \includegraphics[clip,width=0.95\columnwidth]{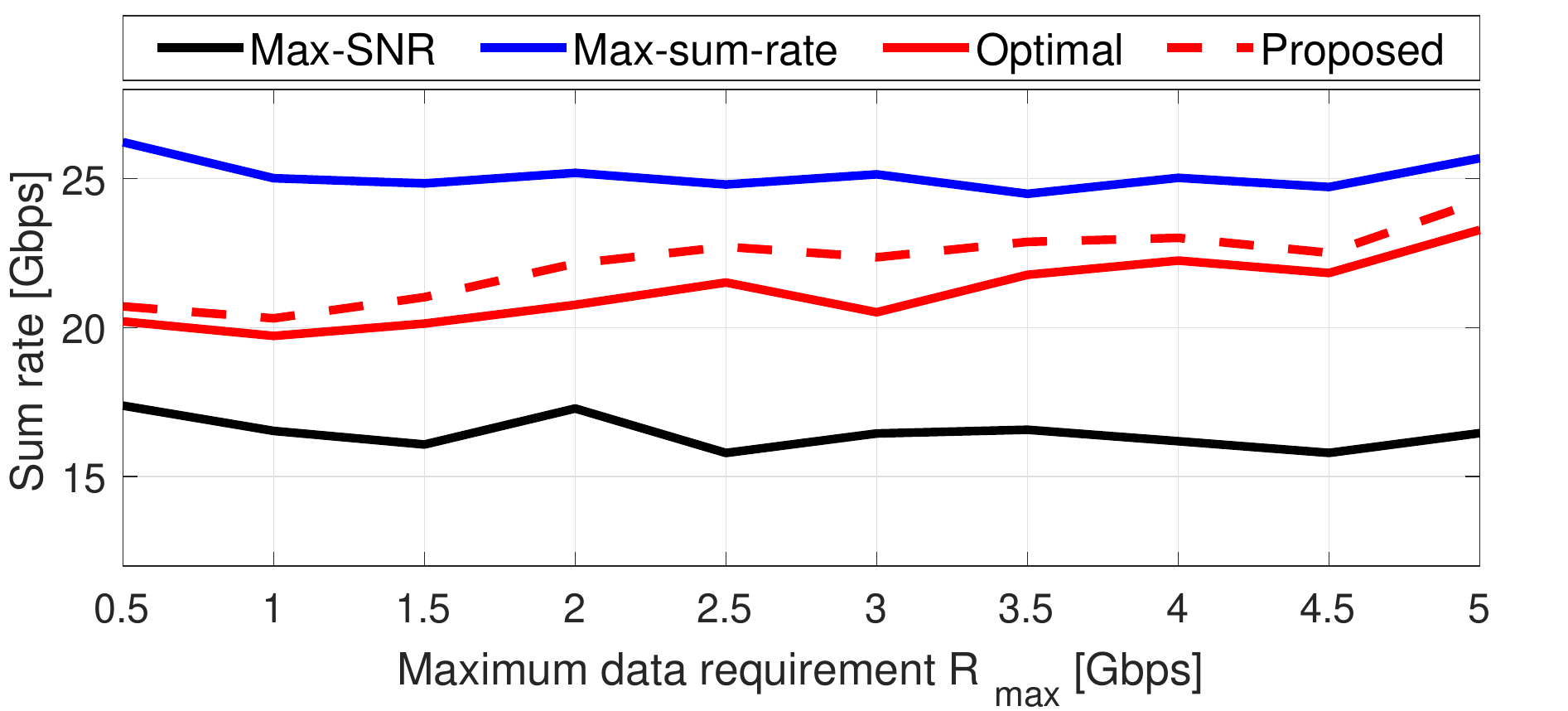}%
}\\
\vspace{0mm}
\end{tabular}
\caption{Comparison in terms of the number of associated UEs with satisfied data requirements and average network sum rate.}
\vspace{-5mm}
\label{fig:mc_comparison}
\end{figure}

%
%
\appendix
\subsection{Integer solution to relaxed problem}
\label{app:integer}

The relaxation of ($\ref{eq:sum_rate}$) can be transformed into a min-cost network flow problem with integer edge capacities and integer vertex supplies/demands.

Let $G=(V,E)$ be a directed graph (network flow), with the source $s$ and the sink $t$. The set of vertices $V$ includes the source $s$, base stations $b$ from $\mathcal{B}$, BS RF chains $j$ from $\mathcal{B}^{\text{NA}}_v$, UE RF chains $i$ from $\mathcal{U}^{\text{NA}}_v$, UEs $u$ from $\mathcal{U}^{\text{NA}}$, and sink $t$. The set $E$ includes edges between:
\begin{itemize}
    \item the source $s$ and the BS $b,~\forall b\in\mathcal{B}$, with edge capacity $c(s,b)=N_{\text{BS}(b)}^{\text{RF}}$ and edge cost $\omega(s,b)=0$,
    \item the BS $b,~\forall b\in\mathcal{B}$, and BS RF chain $j,~\forall j\in\mathcal{B}^{\text{NA}}_v$, with edge capacity $c(b,j)=1$, if $j\rightarrow b$, and $c(b,j)=0$ otherwise, and edge cost $\omega(b,j)=0$,
    \item the BS RF chain $j,~\forall j\in\mathcal{B}^{\text{NA}}_v$, and UE RF chain $i,~\forall i\in\mathcal{U}^{\text{NA}}_v$, with edge capacity $c(j,i)=1$, and edge cost $\omega(j,i)=\frac{1}{1+c_{ij}}$, where $c_{ij}$ are capacities from (\ref{eq:capacity}),
    \item the UE RF chain $i,~\forall i\in\mathcal{U}^{\text{NA}}_v$, and UE $u,~\forall u\in\mathcal{U}$, with edge capacity $c(i,u)=1$, if $i\rightarrow u$, and $c(i,u)=0$ otherwise, and edge cost $\omega(i,u)=0$,
    \item the UE $u,~\forall u\in\mathcal{U}$, and sink $t$, with edge capacity $c(u,t)=N_{\text{UE}}^{\text{RF}}$ and edge cost $\omega(u,t)=0$.
\end{itemize}
The graph $G$ is depicted in Fig.~\ref{fig:graph}.

Let $(m,n)$ be an edge between vertices $m,n\in V$. Let $x(m,n)$ be the flow over the edge $(m,n)$. Let $E'$ be the set of all vertices excluding $s$ and $t$, i.e., $V'=V\setminus \{s,t\}$. The min-cost network flow problem, which is equivalent to the relaxation of (\ref{eq:sum_rate}), can be stated as follows
\begin{maxi!}|s|[3]{(m,n)\in E}{\sum_{(m,n)\in E}x(m,n)\omega(m,n)} {\label{eq:min_cost}}{}
    \addConstraint{\sum_{n:(m,n)\in E}x(m,n)-\sum_{n:(n,m)\in E}x(n,m)}{= 0,~~~\forall m \in V'}{}{\label{eq:min_cost_const_1}}
    \addConstraint{\sum_{n:(s,n)\in E}x(s,n)-\sum_{n:(n,s)\in E}x(n,s)}{= \sum_{b\in\mathcal{B}}N_{\text{RF}(b)}^{\text{BS}}}{}{\label{eq:min_cost_const_2}}
    \addConstraint{\sum_{n:(t,n)\in E}x(t,n)-\sum_{n:(n,t)\in E}x(n,t)}{= -\sum_{b\in\mathcal{B}}N_{\text{RF}(b)}^{\text{BS}}}{}{\label{eq:min_cost_const_3}}
    \addConstraint{0 \leq x(m,n)}{ \leq c(m,n),~~~\forall (m,n)\in E}{}{\label{eq:min_cost_const_4}}
\end{maxi!}
\begin{figure}
    \begin{center}
        \includegraphics[width=0.48\textwidth]{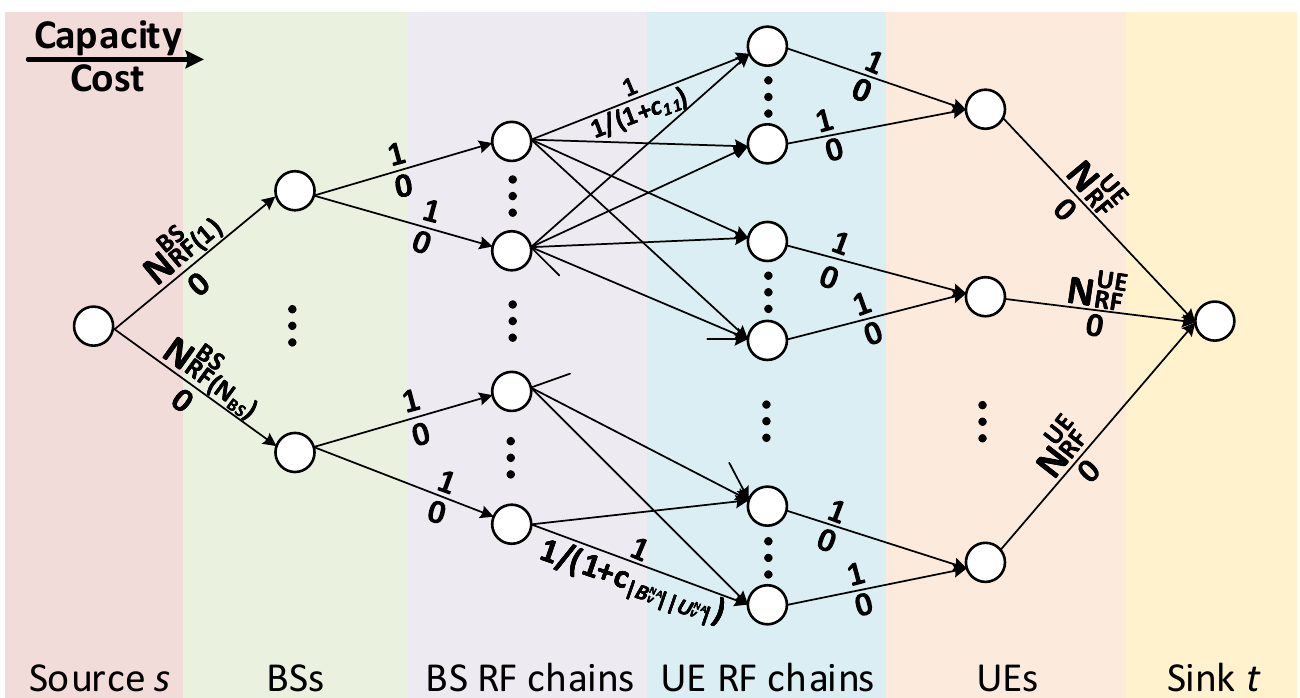}
    \end{center}
    \vspace{-4mm}
    \caption{Directed graph $G$ for the min-cost network flow problem. Zero capacity links are not included in the figure.}
    \vspace{-4mm}
    \label{fig:graph}
\end{figure}
The right-hand sides in (\ref{eq:min_cost_const_1})-(\ref{eq:min_cost_const_3}) represent supply/demand of each vertex in $V$. Both supplies/demands and capacities $c(m,n),\forall(m,n)\in E$, are integers. It was proved that the min-cost network flow problem always have an integer solution when edge capacities and vertex supplies/demands are integer \cite{Ahuja:network}. Thus, there is an integer optimal solution $\mathbf{x}_s^*$ to the relaxation of (\ref{eq:sum_rate}).




%
\bibliographystyle{IEEEtran}
\bibliography{references}

\end{document}